\documentclass[a4paper,12pt]{amsart}

\usepackage{a4wide,graphicx,amssymb,amscd,amsmath,latexsym,amsbsy,amsthm,color,multicol}

\usepackage{marginnote,mathtools,epsfig,enumitem,array,calc,rotating,bm,bbm,mathrsfs} 
\usepackage{subfig}
\usepackage[libertine,cmintegrals,cmbraces,vvarbb]{newtxmath}
\usepackage{booktabs}
\usepackage{multirow}
\usepackage{breqn}
\usepackage{etoolbox}

\usepackage{tikz}

\usetikzlibrary{decorations.markings}
\usetikzlibrary{calc} 
\usetikzlibrary{arrows} 
\usetikzlibrary{patterns} 
\usetikzlibrary{decorations.pathreplacing} 

\usepackage[
   hmarginratio={1:1},     
   vmarginratio={1:1},     
   textwidth=457pt,        
   textheight=630pt,       
   heightrounded,          
]{geometry}

\makeatletter
\let\old@tocline\@tocline
\let\section@tocline\@tocline

\newcommand{\subsection@dotsep}{4.5}
\newcommand{\subsubsection@dotsep}{4.5}
\patchcmd{\@tocline}
  {\hfil}
  {\nobreak
     \leaders\hbox{$\m@th
        \mkern \subsection@dotsep mu\hbox{.}\mkern \subsection@dotsep mu$}\hfill
     \nobreak}{}{}
\let\subsection@tocline\@tocline
\let\@tocline\old@tocline

\patchcmd{\@tocline}
  {\hfil}
  {\nobreak
     \leaders\hbox{$\m@th
        \mkern \subsubsection@dotsep mu\hbox{.}\mkern \subsubsection@dotsep mu$}\hfill
     \nobreak}{}{}
\let\subsubsection@tocline\@tocline
\let\@tocline\old@tocline

\let\old@l@subsection\l@subsection
\let\old@l@subsubsection\l@subsubsection

\def\@tocwriteb#1#2#3{%
  \begingroup
    \@xp\def\csname #2@tocline\endcsname##1##2##3##4##5##6{%
      \ifnum##1>\c@tocdepth
      \else \sbox\z@{##5\let\indentlabel\@tochangmeasure##6}\fi}%
    \csname l@#2\endcsname{#1{\csname#2name\endcsname}{\@secnumber}{}}%
  \endgroup
  \addcontentsline{toc}{#2}%
    {\protect#1{\csname#2name\endcsname}{\@secnumber}{#3}}}%

\newlength{\@tocsectionindent}
\newlength{\@tocsubsectionindent}
\newlength{\@tocsubsubsectionindent}
\newlength{\@tocsectionnumwidth}
\newlength{\@tocsubsectionnumwidth}
\newlength{\@tocsubsubsectionnumwidth}
\newcommand{\settocsectionnumwidth}[1]{\setlength{\@tocsectionnumwidth}{#1}}
\newcommand{\settocsubsectionnumwidth}[1]{\setlength{\@tocsubsectionnumwidth}{#1}}
\newcommand{\settocsubsubsectionnumwidth}[1]{\setlength{\@tocsubsubsectionnumwidth}{#1}}
\newcommand{\settocsectionindent}[1]{\setlength{\@tocsectionindent}{#1}}
\newcommand{\settocsubsectionindent}[1]{\setlength{\@tocsubsectionindent}{#1}}
\newcommand{\settocsubsubsectionindent}[1]{\setlength{\@tocsubsubsectionindent}{#1}}

\renewcommand{\l@section}{\section@tocline{1}{\@tocsectionvskip}{\@tocsectionindent}{\@tocsectionnumwidth}{\@tocsectionformat}}%
\renewcommand{\l@subsection}{\subsection@tocline{1}{\@tocsubsectionvskip}{\@tocsubsectionindent}{\@tocsubsectionnumwidth}{\@tocsubsectionformat}}%
\renewcommand{\l@subsubsection}{\subsubsection@tocline{1}{\@tocsubsubsectionvskip}{\@tocsubsubsectionindent}{\@tocsubsubsectionnumwidth}{\@tocsubsubsectionformat}}%
\newcommand{\@tocsectionformat}{}
\newcommand{\@tocsubsectionformat}{}
\newcommand{\@tocsubsubsectionformat}{}
\expandafter\def\csname toc@1format\endcsname{\@tocsectionformat}
\expandafter\def\csname toc@2format\endcsname{\@tocsubsectionformat}
\expandafter\def\csname toc@3format\endcsname{\@tocsubsubsectionformat}
\newcommand{\settocsectionformat}[1]{\renewcommand{\@tocsectionformat}{#1}}
\newcommand{\settocsubsectionformat}[1]{\renewcommand{\@tocsubsectionformat}{#1}}
\newcommand{\settocsubsubsectionformat}[1]{\renewcommand{\@tocsubsubsectionformat}{#1}}
\newlength{\@tocsectionvskip}
\newcommand{\settocsectionvskip}[1]{\setlength{\@tocsectionvskip}{#1}}
\newlength{\@tocsubsectionvskip}
\newcommand{\settocsubsectionvskip}[1]{\setlength{\@tocsubsectionvskip}{#1}}
\newlength{\@tocsubsubsectionvskip}
\newcommand{\settocsubsubsectionvskip}[1]{\setlength{\@tocsubsubsectionvskip}{#1}}

\patchcmd{\tocsection}{\indentlabel}{\makebox[\@tocsectionnumwidth][l]}{}{}
\patchcmd{\tocsubsection}{\indentlabel}{\makebox[\@tocsubsectionnumwidth][l]}{}{}
\patchcmd{\tocsubsubsection}{\indentlabel}{\makebox[\@tocsubsubsectionnumwidth][l]}{}{}

\newcommand{\@sectypepnumformat}{}
\renewcommand{\contentsline}[1]{%
  \expandafter\let\expandafter\@sectypepnumformat\csname @toc#1pnumformat\endcsname%
  \csname l@#1\endcsname}
\newcommand{\@tocsectionpnumformat}{}
\newcommand{\@tocsubsectionpnumformat}{}
\newcommand{\@tocsubsubsectionpnumformat}{}
\newcommand{\setsectionpnumformat}[1]{\renewcommand{\@tocsectionpnumformat}{#1}}
\newcommand{\setsubsectionpnumformat}[1]{\renewcommand{\@tocsubsectionpnumformat}{#1}}
\newcommand{\setsubsubsectionpnumformat}[1]{\renewcommand{\@tocsubsubsectionpnumformat}{#1}}
\renewcommand{\@tocpagenum}[1]{%
  \hfill {\mdseries\@sectypepnumformat #1}}

\let\oldappendix\appendix
\renewcommand{\appendix}{%
  \leavevmode\oldappendix%
  \addtocontents{toc}{%
    \protect\settowidth{\protect\@tocsectionnumwidth}{\protect\@tocsectionformat\sectionname\space}%
    \protect\addtolength{\protect\@tocsectionnumwidth}{2em}}%
}
\makeatother

\makeatletter
\settocsectionnumwidth{2em}
\settocsubsectionnumwidth{2.5em}
\settocsubsubsectionnumwidth{3em}
\settocsectionindent{1pc}%
\settocsubsectionindent{\dimexpr\@tocsectionindent+\@tocsectionnumwidth}%
\settocsubsubsectionindent{\dimexpr\@tocsubsectionindent+\@tocsubsectionnumwidth}%
\makeatother

\settocsectionvskip{0pt}
\settocsubsectionvskip{-5pt}
\settocsubsubsectionvskip{-5pt}

\settocsectionformat{\bfseries}
\settocsubsectionformat{\mdseries}
\settocsubsubsectionformat{\mdseries}
\setsectionpnumformat{\bfseries}
\setsubsectionpnumformat{\mdseries}
\setsubsubsectionpnumformat{\mdseries}

\let\oldtableofcontents\tableofcontents
\renewcommand{\tableofcontents}{%
  \vspace*{-\linespacing}
  \oldtableofcontents}

\numberwithin{equation}{section}
\theoremstyle{plain}
 
\newtheorem{thm}{Theorem}[section]

\newtheorem{defi}[thm]{Definition}
\newtheorem{lem}[thm]{Lemma}

\theoremstyle{remark}
\newtheorem{rema}[thm]{Remark}

\newcommand{\Z}{\mathbb{Z}}

\newcommand{\C}{\mathbb{C}}

\newcommand{\ii}{\mathrm{i}}
\newcommand{\ket}[1]{\left|#1\right\rangle}      
\newcommand{\braket}[2]{\left\langle #1 \right. | \left.  #2 \right\rangle}

\newcommand{\hypref}[2]{\ifx\href\asklfhas #2\else\href{#1}{#2}\fi}
\newcommand{\Secref}[1]{Section~\ref{#1}}

\renewcommand{\eqref}[1]{(\ref{#1})}

\def\[{\begin{equation}}
\def\]{\end{equation}}
\def\<{\begin{eqnarray}}
\def\>{\end{eqnarray}}

\title[]{Domain-wall boundaries through non-diagonal twists \\ in the six-vertex model}
\author{W. Galleas}

\address{Institut f\"ur Theoretische Physik, Eidgen\"ossische Technische Hochschule Z\"urich, Wolfgang-Pauli-Strasse 27, 8093 Z\"urich, Switzerland}
\email{galleasw@phys.ethz.ch}

\thanks{The work of W.G. is partially supported by the Swiss National Science Foundation through the NCCR SwissMAP}

\begin{document}

\begin{abstract}
In this work we elaborate on a previous result relating the partition function of the six-vertex model with domain-wall boundary conditions to
eigenvalues of a transfer matrix. More precisely, we express the aforementioned partition function as a determinant of a matrix with entries
being eigenvalues of the anti-periodic six-vertex model's transfer matrix. 
\end{abstract}

\maketitle

\tableofcontents

\section{Introduction} \label{sec:INTRO}

The study of lattice models of Statistical Mechanics received a large impulse with the advent of Kramers and Wannier transfer matrix technique \cite{Kramers_1941a, Kramers_1941b}. This method aims to evaluate the partition function of a given model in closed form; and it does that by recasting the model's partition function in terms of eigenvalues of an operator known as \emph{transfer matrix}. It is worth remarking that this method has been notably successful in two-dimensions and it has led, for instance, to the exact solution of the two-dimensional Ising model \cite{Onsager_1944} and the eight-vertex model \cite{Baxter_1972}.
Although the idea behind this method seems to be quite general, its practical implementation still depends on the boundary conditions assumed for the model. For instance, two-dimensional vertex models with \emph{toroidal boundary conditions} admit such reformulation comfortably as the associated transfer matrices are usually diagonalizable. On the other hand, transfer matrices associated to the case of \emph{domain-wall boundaries} are not diagonalizable in general, although they can still be made triangular; and this issue prevents the reformulation of such models as eigenvalue problems in the strict sense. 

The six-vertex model with domain-wall boundary conditions was introduced by Korepin as a tool for studying scalar products of Bethe vectors \cite{Korepin_1982}. 
However, in the course of the years this model has acquired life on its own and was ultimately responsible for clarifying fundamental concepts of Statistical Mechanics related to the role played by boundary conditions in the thermodynamical limit \cite{Justin_2000}. In addition to that, connections between the six-vertex model with domain-wall boundaries and the theory of special functions \cite{Warnaar_2008}, enumerative combinatorics \cite{Kuperberg_1996} and matrix models \cite{Justin_2000a} have also been reported in the literature to date. 

As previously remarked, vertex models with domain-wall boundaries can offer fundamental obstacles for their reformulation as strict eigenvalue problems; and this is in clear contrast to the case with toroidal boundary conditions. Nevertheless, as for the six-vertex model and elliptic generalizations one can still express such partition functions as a single determinant \cite{Izergin_1987, Galleas_2016b}; hence, in those cases the partition function can still be regarded as a product of eigenvalues.
The first determinantal representation for the partition function of the six-vertex model with domain-wall boundaries was found by Izergin
in \cite{Izergin_1987}. In particular, Izergin obtained such representation by means of an \emph{educated guess} which can be shown to satisfy a set of properties previously derived by Korepin in \cite{Korepin_1982}. Those properties, in their turn, can be shown to characterize the model's partition function uniquely; ensuring the validity of Izergin's formula.

Several alternative representations have also been found for the partition function of the six-vertex model with domain-wall boundaries. For instance, it can be expressed as multiple contour integrals \cite{deGier_Galleas_2011, Galleas_2012, Galleas_2013} or as continuous families of determinants \cite{Galleas_2016, Galleas_2016b}.
As far as the determinants obtained in \cite{Galleas_2016, Galleas_2016b} are concerned, it is important to remark they differ significantly from the one obtained by Izergin. Moreover, the determinant formulae of \cite{Galleas_2016, Galleas_2016b} are obtained from a constructive approach and 
no \emph{educated guess} is required at any point. 

Turning our attention back to the case with toroidal boundary conditions, it is important to remark that even though those cases can be tackled through 
the transfer matrix technique, there is still no guarantee that the associated eigenvalue problem can indeed be solved.
At this stage it is also important to stress that we are referring to the case with \emph{twisted boundary conditions} described in \cite{deVega_1984} also as
toroidal boundaries. Strictly speaking, the case considered in \cite{deVega_1984} describes deviations from periodic boundary conditions but, nevertheless,
the physical system under consideration is still defined on a torus \footnote{As for the boundary conditions, our use of the term \emph{twist} should not be confused with the term \emph{screw} used in \cite{Kramers_1941a}.}.
Such deviations are characterized by a certain matrix $\mathcal{G}$ which is then required to satisfy a few properties in order to preserve the integrability of the bulk vertex model. As for the six-vertex model the matrix $\mathcal{G}$ can be either purely diagonal or off-diagonal as shown in \cite{deVega_1984}. In the language of quantum hamiltonians, the off-diagonal case gives rise to \textsc{xxz} spin-chain with anti-periodic boundary conditions and we shall also use the terminology \emph{anti-periodic} when refereeing to the associated six-vertex model. 

As far as the transfer matrix's eigenvalue problem is concerned, the case where $\mathcal{G}$ is diagonal poses no obstacle for a standard Bethe ansatz analysis. The off-diagonal case, on the other hand, has been a subject of discussion for a long time \cite{Baxter_1995, Ribeiro_Galleas_2003, WengLi_2013, Galleas_Twists}. In particular, in the work \cite{Galleas_Twists} we have obtained a formula expressing the partition function of the six-vertex model with domain-wall boundaries in terms of eigenvalues of the anti-periodic six-vertex model's transfer matrix. The existence of such formula is quite unexpected since it relates systems with different boundary conditions. However, the formula obtained in \cite{Galleas_Twists} is far from simple and this issue apparently makes any further analysis out of reach. This is precisely the point where we intend to shed some light in the present paper. More precisely, here we intend to show that the relation firstly observed in \cite{Galleas_Twists} can be written as a neat determinant. 

This paper is organized as follows. In \Secref{sec:FUN} we describe the system of functional equations which will lead us to the main result of the present work. We also use \Secref{sec:FUN} to introduce conventions which will be employed throughout this paper. The analysis of the aforementioned functional equations is then presented in \Secref{sec:SOL}. In particular, we first describe our methodology for small lattice lengths for the sake of clarity before presenting the general case. Some consequences of our results are then discussed in \Secref{sec:LAM} and \Secref{sec:CONCL} is left for concluding remarks.

\section{Functional equations} \label{sec:FUN}

This is a preliminary section aimed to introduce conventions employed throughout this paper and describe the system of functional equations which will lead to the anticipated formula for the partition function of the six-vertex model with domain-wall boundaries. Since the present paper is largely based on results previously obtained in \cite{Galleas_Twists}, we shall deliberately omit details and restrict ourselves to presenting only the definitions and results which will be required in our present analysis.

\begin{defi}[Parameters]
Let $L \in \Z_{>0}$ be the lattice length and $\lambda_j \in \C$ for $1 \leq j \leq L$ be refereed to as spectral parameters. Also, write $\mu_j \in \C$ for the so called
inhomogeneity parameters on the same interval $1 \leq j \leq L$. The anisotropy parameter will be denoted by $\gamma \in \C$. The parameters $L$, $\mu_j$ and $\gamma$ are fixed.
\end{defi}

Among the main players in this work we have the eigenvalues $\Lambda \colon \C \to \C$ of the transfer matrix $T \colon \C \to \mathrm{End}\left( (\C^2)^{\otimes L} \right)$ described in \cite{Galleas_Twists}. They will share the stage with the partition function $Z \colon \C^L \to \C$ of the six-vertex model with domain-wall boundaries. Using the canonical formulation of (generalized) toroidal vertex models in terms of transfer matrices, the operator $T$ considered here builds the partition function of the six-vertex model with anti-periodic boundary conditions \cite{deVega_1984}. In particular, it enjoys Baxter's commutativity property $\left[ T(\lambda_1) , T(\lambda_2)  \right] = 0 \; \; \forall \lambda_1, \lambda_2 \in \C$ and its eigenvalue problem reads
$T(\lambda) \ket{\Psi} = \Lambda(\lambda) \ket{\Psi}$ for $\ket{\Psi} \in \mathrm{span}\left( (\C^2)^{\otimes L}  \right)$.

As for the partition function $Z$, it is a multivariate function satisfying a number of well known properties. In particular, the symmetry property 
\begin{equation*}
Z( \dots , \lambda_i , \dots, \lambda_j , \dots ) = Z( \dots , \lambda_j , \dots , \lambda_i , \dots )
\end{equation*}
plays a fundamental role in the characterization of $Z$. 
The ultimate goal of the present paper is to find a neat relation between $Z$ and $\Lambda$. This relation will be intermediated by a sequence of 
symmetric functions $\mathcal{F}_n \colon \C^n \to \C$ $(0 \leq n < L)$ satisfying a system of functional equations. In order to describe such system of equations we first need to introduce extra conventions.

\begin{defi}[Variables sets] \label{VAR}
Write $X^{p,q} \coloneqq \{ \lambda_k \in \C \mid p \leq k \leq q \}$ and additionally define $X^{p,q}_i \coloneqq X^{p,q} \backslash \{ \lambda_i \}$. We shall also use $X^{p,q}_{i,j} \coloneqq X^{p,q} \backslash \{ \lambda_i , \lambda_j \}$.
\end{defi}

\begin{rema}
The set $X^{p,q}$ and its reduced definitions will be mainly used to denote the arguments of multivariate symmetric functions. For instance, using this notation one can simply write $Z(\lambda_1, \lambda_2, \dots , \lambda_L) = Z( X^{1,L} )$.
\end{rema}

Throughout this work we will also need to consider non-symmetric multivariate functions and for that it is convenient to introduce the following vectors (ordered sets) to denote the argument of relevant functions.

\begin{defi}[Ordered sets] \label{oVAR}
Consider the vector $\mathcal{O} \left( X^{p,q} \right)  \coloneqq \left( \lambda_p, \lambda_{p+1}, \dots , \lambda_q  \right)$ with $ p \leq q$ obtained from 
$X^{p,q}$ by fixing the ordering of its elements in increasing order of variables $\lambda_i$.
We then write 
\[
\mathcal{O} \left( X^{p,q} \right)_i  \coloneqq ( \lambda_p, \lambda_{p+1}, \dots , \underbrace{\widehat{\lambda}}_{\text{i-th}}, \dots,  \lambda_q  )
\]
for the vector obtained from $\mathcal{O} \left( X^{p,q} \right)$ by removing its $i$-th component; and $\mathcal{O} \left( X^{p,q} \right)_{i,j} \coloneqq \left( \mathcal{O} \left( X^{p,q} \right)_i \right)_j$ for the one obtained by removing the $j$-th component of $\mathcal{O} \left( X^{p,q} \right)_i$.
\end{defi}

\begin{defi}[Auxiliary functions]
Assuming the same conventions employed in \cite{Galleas_Twists}, we write $a(\lambda) \coloneqq \sinh{(\lambda + \gamma)}$, $b(\lambda) \coloneqq \sinh{(\lambda)}$
and $c(\lambda) \coloneqq \sinh{(\gamma)}$ for $\lambda \in \C$. Then for $n \in \{1, 2, \dots , L+1\}$ and $1 \leq i < j \leq n$ we define the coefficient functions
$M_i^{(n)} \colon \C^{n+1} \to \C$ and $N_{j, i}^{(n)} \colon \C^{n+1} \to \C$ as 
\< \label{MN}
&& M_i^{(n)} (\lambda_0 , \lambda_1 , \dots , \lambda_n) \coloneqq \frac{c(\lambda_i - \lambda_0)}{b(\lambda_i - \lambda_0)} \prod_{\substack{k=1 \\ k \neq i}}^n \frac{a(\lambda_i - \lambda_k)}{b(\lambda_i - \lambda_k)} \frac{a(\lambda_k - \lambda_0)}{b(\lambda_k - \lambda_0)} \prod_{l=1}^L a(\lambda_0 - \mu_l) b(\lambda_i - \mu_l)
\nonumber \\
&& \qquad \qquad \quad \qquad \qquad + \frac{c(\lambda_0 - \lambda_i)}{b(\lambda_0 - \lambda_i)} \prod_{\substack{k=1 \\ k \neq i}}^n \frac{a(\lambda_0 - \lambda_k)}{b(\lambda_0 - \lambda_k)} \frac{a(\lambda_k - \lambda_i)}{b(\lambda_k - \lambda_i)} \prod_{l=1}^L a(\lambda_i - \mu_l) b(\lambda_0 - \mu_l) \nonumber \\
&& N_{j,i}^{(n)} (\lambda_0 , \lambda_1 , \dots , \lambda_n) \coloneqq \nonumber \\
&& \frac{c(\lambda_j - \lambda_0)}{a(\lambda_j - \lambda_0)} \frac{c(\lambda_0 - \lambda_i)}{a(\lambda_0 - \lambda_i)} \frac{a(\lambda_j - \lambda_i)}{b(\lambda_j - \lambda_i)} \prod_{\substack{k=0 \\ k \neq i, j}}^n \frac{a(\lambda_j - \lambda_k)}{b(\lambda_j - \lambda_k)} \frac{a(\lambda_k - \lambda_i)}{b(\lambda_k - \lambda_i)} \prod_{l=1}^L a(\lambda_i - \mu_l) b(\lambda_j - \mu_l) \nonumber \\
&& + \frac{c(\lambda_i - \lambda_0)}{a(\lambda_i - \lambda_0)} \frac{c(\lambda_0 - \lambda_j)}{a(\lambda_0 - \lambda_j)} \frac{a(\lambda_i - \lambda_j)}{b(\lambda_i - \lambda_j)} \prod_{\substack{k=0 \\ k \neq i, j}}^n \frac{a(\lambda_i - \lambda_k)}{b(\lambda_i - \lambda_k)} \frac{a(\lambda_k - \lambda_j)}{b(\lambda_k - \lambda_j)} \prod_{l=1}^L a(\lambda_j - \mu_l) b(\lambda_i - \mu_l) \; . \nonumber  \\
\>
\end{defi}

\begin{rema}
Taking into account the conventions introduced in Definition \ref{oVAR}, one can also write $M_i^{(n)} (\lambda_0 , \lambda_1 , \dots , \lambda_n) = M_i^{(n)} (\mathcal{O} \left( X^{0,n} \right))$ and $N_{j,i}^{(n)} (\lambda_0 , \lambda_1 , \dots , \lambda_n) = N_{j,i}^{(n)} (\mathcal{O} \left( X^{0,n} \right))$. 
\end{rema}

We have now gathered all the ingredients which will be required to formulate the aforementioned system of functional equations. A detailed analysis of such equations will be presented in the following sections.

\begin{lem}[Functional equations] \label{FZL}
$\exists \; \mathcal{F}_n \colon \C^n \to \C$ satisfying the system of equations
\< \label{FN}
\Lambda (\lambda_0) \mathcal{F}_n (X^{1,n}) &=& \mathcal{F}_{n+1} (X^{0,n}) + \sum_{1 \leq i \leq n} M_i^{(n)} \mathcal{F}_{n-1} (X^{1,n}_i) \nonumber\\
&& + \; \sum_{1 \leq i < j \leq n} N_{j,i}^{(n)} \mathcal{F}_{n-1} (X^{0,n}_{i,j})
\>
for  $0 \leq n \leq L-1$. In \eqref{FN} $M_i^{(n)} = M_i^{(n)} (\lambda_0 , \lambda_1 , \dots , \lambda_n)$ and $N_{j,i}^{(n)}  = N_{j,i}^{(n)} (\lambda_0 , \lambda_1 , \dots , \lambda_n)$ as defined in \eqref{MN}; and we recall $\mathcal{F}_n$ is a symmetric function. Also, we identify
$\mathcal{F}_L (X^{1,L}) = Z(X^{1,L}) \bar{\mathcal{F}}_0$ with $\bar{\mathcal{F}}_0$ a complex parameter.
\end{lem}

\begin{rema}
The partition function $Z$ only appears in the system \eqref{FN} for $n=L-1$; whilst $\Lambda$ is present in each equation of the system. In this way, one can expect that the sequence of functions $\{ \mathcal{F}_n \}_{0 \leq n \leq L-1}$ will be required to build a bridge between $Z$ and $\Lambda$.
\end{rema}

\section{The relation $\Lambda \rightarrow Z$} \label{sec:SOL}

One important aspect of the system of equations \eqref{FN} is that its solution depends mostly on the equations' structure rather than on the particular form of the coefficients  $M_i^{(n)}$ and  $N_{j,i}^{(n)}$. That is what we intend to demonstrate in this section and it is worth remarking that this feature has been previously exploited in \cite{Galleas_2016a, Galleas_2016b, Galleas_2016c} for similar equations.
This section will be necessarily technical but, before we go through the details, it is important to first precise what we mean by \emph{solution}.
In \eqref{FN} we have assumed the function $\Lambda(\lambda)$ is given and the unknown functions are then the set $\{ \mathcal{F}_n \}_{1 \leq n \leq L}$, keeping in mind $\mathcal{F}_L \sim Z$. Therefore, we would like to express the functions $\mathcal{F}_n$ in terms of $\Lambda$, $M_i^{(n)}$ and $N_{j,i}^{(n)}$ in the simplest possible way.
Although it is doable through our approach, here we will not attempt to presenting a formula for each function $\mathcal{F}_n$ since their interpretation in terms of vertex models is still unclear to us. However, the function $\mathcal{F}_L$ is essentially the partition function $Z$ and we will focus on that case.

The previous works \cite{Galleas_2016a, Galleas_2016b, Galleas_2016c} have unveiled some useful attributes of functional equations with structure similar to \eqref{FN}. For instance, they not only present themselves as functional equations but they also exhibit the structure of linear algebraic equations. This feature has been exploited in the aforementioned works where we have shown that determinantal solutions can be readily written down with little (or almost none) dependence on the explicit form of its coefficients.
Here we intend to show \eqref{FN} can be tackled along the same lines and, at the end of the day, we will find a solutions
$\mathcal{F}_L$ (or  $Z$) as the determinant of a matrix with entries depending solely on $\Lambda$, $M_i^{(n)}$ and $N_{j,i}^{(n)}$. 
In order to explain our methodology we shall first address the particular cases $L = 2 , 3 , 4$ before describing the general strategy.

\subsection{Case $L=2$} \label{sec:L2}

This is the first non-trivial instance of the system of equations \eqref{FN}. In that case we have
\< \label{FN2}
\Lambda(\lambda_0) \mathcal{F}_0 &=& \mathcal{F}_1 (\lambda_0) \nonumber \\
\Lambda(\lambda_0) \mathcal{F}_1 (\lambda_1) &=& Z(\lambda_0, \lambda_1) \bar{\mathcal{F}}_0 + M_1^{(1)}(\lambda_0, \lambda_1) \mathcal{F}_0 \; . 
\>
Then, recalling that $\lambda_j$ is a generic complex variable, we can rewrite \eqref{FN2} in matricial form as
\<
\begin{pmatrix} 
\Lambda (\lambda_1) & -1 \\ M_1^{(1)}(\lambda_0, \lambda_1) & - \Lambda(\lambda_0) 
\end{pmatrix}
\begin{pmatrix}
\mathcal{F}_0 \\  \mathcal{F}_1 (\lambda_1) 
\end{pmatrix} = 
\begin{pmatrix}
0 \\ - Z(\lambda_0, \lambda_1) \bar{\mathcal{F}}_0 
\end{pmatrix} \; ,
\>
which can be readily solved using Cramer's rule. In this way we find
\<
\mathcal{F}_0 = - \frac{Z(\lambda_0, \lambda_1) \bar{\mathcal{F}}_0}{ \mathrm{det} \begin{pmatrix} 
\Lambda (\lambda_1) & -1 \\ M_1^{(1)}(\lambda_0, \lambda_1) & - \Lambda(\lambda_0) 
\end{pmatrix}} \; ,
\>
or equivalently,
\< \label{Z2}
Z(\lambda_0, \lambda_1) = \kappa_0 \; \mathrm{det} \begin{pmatrix} 
\Lambda (\lambda_1) & 1 \\ M_1^{(1)}(\lambda_0, \lambda_1) & \Lambda(\lambda_0) \end{pmatrix}
\>
with $\kappa_0 \coloneqq  \mathcal{F}_0 / \bar{\mathcal{F}}_0$. Formula \eqref{Z2} then fulfills the goal of expressing $Z$ as a single determinant involving solely any eigenvalue $\Lambda$ and the function $M_1^{(1)}$. It is also important to point out that formula \eqref{Z2} is a direct consequence of the structure of \eqref{FN2}, and that its derivation has not made use of the explicit form of the function $M_1^{(1)}$ at any point.
Moreover, as for $L=2$ one can verify the function $M_1^{(1)}$ entering formula \eqref{Z2} simplifies to 
\<
M_1^{(1)} (\lambda_0, \lambda_1) &=& -  \frac{\sinh{(\gamma)}^2}{2} \left[ \cosh{(\lambda_0 - \lambda_1 + \gamma)} + \cosh{(\lambda_1 - \lambda_0 + \gamma)} \right. \nonumber \\
&& \qquad \qquad \qquad \qquad \quad \left. - \; \sum_{j=1}^2 \cosh{(\lambda_0 + \lambda_1 + \gamma - 2 \mu_j)} \right] .
\>
Although formula \eqref{MN} for $M_1^{(1)}$ seems to exhibit a simple pole when $\lambda_0, \lambda_1 \to \lambda$, we can see this is not the case for $L=2$. In this way we simply have 
\[ \label{Z2homo}
Z(\lambda, \lambda) = \kappa_0 \left[ \Lambda(\lambda)^2 -  M_1^{(1)} (\lambda, \lambda)   \right] 
\]
in the \emph{homogeneous limit} $\lambda_0, \lambda_1 \to \lambda$.

Both formulae \eqref{Z2} and \eqref{Z2homo} can not be regarded as explicit expressions since one still need to input an eigenvalue $\Lambda$ of the transfer matrix $T$. The spectrum of the latter comprises four eigenvalues for $L=2$; which can be obtained from the direct diagonalization of $T$. In Table \ref{tab:L2} we have listed such eigenvalues and used them to compare formula \eqref{Z2} with the explicit evaluation of $Z$ from its definition (sum over configurations). More precisely, we have compared \eqref{Z2} with
\[
Z(\lambda_0, \lambda_1) = c^2 \left[ b(\lambda_0 - \mu_1) b(\lambda_1 - \mu_2) + a(\lambda_0 - \mu_2) a(\lambda_1 - \mu_1)  \right]
\]
which shows our determinantal formula is fulfilled with the respective parameter $\kappa_0$ also given in Table \ref{tab:L2}. 

\begin{table}[h]
\caption{Eigenvalues $\Lambda$ for $L=2$.}
\label{tab:L2}
\begin{center}
\begin{tabular}{|c|c|}
\hline
$\kappa_0$ & $\Lambda(\lambda)$ \\
\hline
1 & $\sqrt{2} \sinh{(\gamma)} \sqrt{\cosh{(\gamma)}+\cosh{(\mu_1 - \mu_2)}} \sinh{\left(\frac{1}{2}(\gamma-\mu_1 -\mu_2 + 2 \lambda)\right)}$ \\
\hline
1 & $-\sqrt{2} \sinh{(\gamma)} \sqrt{\cosh{(\gamma)}+\cosh{(\mu_1 - \mu_2)}} \sinh{\left(\frac{1}{2}(\gamma-\mu_1 -\mu_2 + 2 \lambda)\right)}$  \\
\hline
-1 & $\ii \sqrt{2} \sinh{(\gamma)} \sqrt{\cosh{(\gamma)}-\cosh{(\mu_1 - \mu_2)}} \cosh{\left(\frac{1}{2}(\gamma-\mu_1 -\mu_2 + 2 \lambda)\right)}$  \\
\hline
-1 & $-\ii \sqrt{2} \sinh{(\gamma)} \sqrt{\cosh{(\gamma)}-\cosh{(\mu_1 - \mu_2)}} \cosh{\left(\frac{1}{2}(\gamma-\mu_1 -\mu_2 + 2 \lambda)\right)}$ \\
\hline
\end{tabular}
\end{center}
\end{table}

\subsection{Case $L=3$} \label{sec:L3}

We then proceed with a detailed analysis of the case $L=3$. In that case \eqref{FN} comprises the following three equations:
\< \label{FN3}
\Lambda(\lambda_0) \mathcal{F}_0 &=& \mathcal{F}_1 (\lambda_0) \nonumber \\
\Lambda(\lambda_0) \mathcal{F}_1 (\lambda_1) &=& \mathcal{F}_2 (\lambda_0, \lambda_1)  + M_1^{(1)}(\lambda_0, \lambda_1) \mathcal{F}_0 \nonumber \\
\Lambda(\lambda_0) \mathcal{F}_2 (\lambda_1, \lambda_2) &=& Z(\lambda_0, \lambda_1 , \lambda_2) \; \bar{\mathcal{F}}_0 + M_1^{(2)}(\lambda_0, \lambda_1 , \lambda_2) \; \mathcal{F}_1 (\lambda_2) \nonumber \\
&& + \; M_2^{(2)}(\lambda_0, \lambda_1 , \lambda_2) \; \mathcal{F}_1 (\lambda_1) + N_{2,1}^{(2)}(\lambda_0, \lambda_1 , \lambda_2) \; \mathcal{F}_1 (\lambda_0) \; .
\>

The generalization of the analysis performed in the case $L=2$ for the system of equations \eqref{FN3} is not so straightforward. 
This is mainly due to the fact that one first needs to select appropriate equations and variables among several possibilities encoded in \eqref{FN3}. 
In order to choose suitable variables we recall again that $\lambda_j \in \C$ is generic and first consider the following sub-system of equations already written in matricial form,
\< \label{FN3a}
\begin{pmatrix} 
\Lambda (\lambda_2) & -1 & 0 & 0 \\ 
\Lambda (\lambda_1) & 0 & -1 & 0 \\
\Lambda (\lambda_0) & 0 & 0 & -1 \\
0 & M_1^{(2)} & M_2^{(2)} & N_{2,1}^{(2)}
\end{pmatrix}
\begin{pmatrix}
\mathcal{F}_0 \\  \mathcal{F}_1 (\lambda_2) \\  \mathcal{F}_1 (\lambda_1) \\  \mathcal{F}_1 (\lambda_0)
\end{pmatrix} = 
\begin{pmatrix}
0 \\ 0 \\ 0 \\ \Lambda(\lambda_0) \mathcal{F}_2 (\lambda_1, \lambda_2) -  Z(\lambda_0, \lambda_1 , \lambda_2) \; \bar{\mathcal{F}}_0
\end{pmatrix} \; . \nonumber \\
\>
In \eqref{FN3a} the first three rows corresponds to the first equation of \eqref{FN3} with different spectral parameters. The last row is simply the last equation of \eqref{FN3}.

Now one can use Cramer's rule to solve \eqref{FN3a} for $\mathcal{F}_0$. By doing so we obtain the relation
\< \label{FN3A}
\Lambda(\lambda_0) \mathcal{F}_2 (\lambda_1, \lambda_2) -  Z(\lambda_0, \lambda_1 , \lambda_2) \; \bar{\mathcal{F}}_0 = 
\mathcal{F}_0 \; \mathrm{det} \begin{pmatrix} 
\Lambda (\lambda_2) & -1 & 0 & 0 \\ 
\Lambda (\lambda_1) & 0 & -1 & 0 \\
\Lambda (\lambda_0) & 0 & 0 & -1 \\
0 & M_1^{(2)} & M_2^{(2)} & N_{2,1}^{(2)}
\end{pmatrix} \; ,
\>
which is close to the type of formula we are after but not quite there yet. In order to complete our formula we still need a suitable expression for 
$\mathcal{F}_2$. Such formula can be found along the same lines; more precisely, by replacing the last row of \eqref{FN3a} by the second equation in
\eqref{FN3} with adjusted spectral parameters. Hence, we consider the sub-system of equations
\< \label{FN3b}
\begin{pmatrix} 
\Lambda (\lambda_2) & -1 & 0 & 0 \\ 
\Lambda (\lambda_1) & 0 & -1 & 0 \\
\Lambda (\lambda_0) & 0 & 0 & -1 \\
- M_1^{(1)}(\lambda_1 , \lambda_2) & \Lambda(\lambda_1) & 0 & 0 
\end{pmatrix}
\begin{pmatrix}
\mathcal{F}_0 \\  \mathcal{F}_1 (\lambda_2) \\  \mathcal{F}_1 (\lambda_1) \\  \mathcal{F}_1 (\lambda_0)
\end{pmatrix} = 
\begin{pmatrix}
0 \\ 0 \\ 0 \\ \mathcal{F}_2 (\lambda_1, \lambda_2) 
\end{pmatrix}  \nonumber \\
\>
whose solution for $\mathcal{F}_0$ gives
\< \label{FN3B}
\mathcal{F}_2 (\lambda_1, \lambda_2) = \mathcal{F}_0 \; \mathrm{det} \begin{pmatrix} 
\Lambda (\lambda_2) & -1 & 0 & 0 \\ 
\Lambda (\lambda_1) & 0 & -1 & 0 \\
\Lambda (\lambda_0) & 0 & 0 & -1 \\
- M_1^{(1)}(\lambda_1 , \lambda_2) & \Lambda(\lambda_1) & 0 & 0 
\end{pmatrix} \; .
\>
The determinants appearing in \eqref{FN3A} and \eqref{FN3B} only differ by the last row and they can be added using elementary properties of determinants. In this way we are left with the formula
\< \label{Z3}
Z(\lambda_0, \lambda_1 , \lambda_2) = \kappa_0 \; \mathrm{det} \begin{pmatrix} 
\Lambda (\lambda_2) & -1 & 0 & 0 \\ 
\Lambda (\lambda_1) & 0 & -1 & 0 \\
\Lambda (\lambda_0) & 0 & 0 & -1 \\
- M_1^{(1)}(\lambda_1 , \lambda_2) \Lambda(\lambda_0)  & \Lambda(\lambda_0) \Lambda(\lambda_1) - M_1^{(2)}  & - M_2^{(2)} & - N_{2,1}^{(2)}
\end{pmatrix} \; , \nonumber \\
\>
with $M_i^{(2)} = M_i^{(2)} (\lambda_0 , \lambda_1 , \lambda_2)$ and $N_{2,1}^{(2)} = N_{2,1}^{(2)} (\lambda_0 , \lambda_1 , \lambda_2)$.

Similarly to the analysis performed for $L=2$, we also would like to compare the determinantal expression \eqref{Z3} with the explicit evaluation of $Z$ from its definition. For that we first need to compute the eigenvalues of the transfer matrix $T$ for $L=3$. In order to avoid cumbersome expressions we then set $\mu_j = 0$ and present the aforementioned eigenvalues in Table \ref{tab:L3}. Next we need an explicit expression for $Z$ computed directly from its definition. As for $L=3$ and $\mu_j = 0$ we have
\< \label{Z3r}
Z(\lambda_0, \lambda_1 , \lambda_2) &=& c^3 \left[c^2 a\left(\lambda _0\right) a\left(\lambda _2\right) b\left(\lambda_0\right) b\left(\lambda_2\right)+a\left(\lambda_0\right) a\left(\lambda_1\right) b\left(\lambda_0\right) b\left(\lambda_1\right) b\left(\lambda_2\right){}^2 \right. \nonumber \\
&& \left. + \; a\left(\lambda_0\right) a\left(\lambda_1\right) a\left(\lambda_2\right){}^2 b\left(\lambda_0\right) b\left(\lambda_1\right)+a\left(\lambda_1\right) a\left(\lambda_2\right) b\left(\lambda_0\right){}^2 b\left(\lambda_1\right) b\left(\lambda_2\right) \right. \nonumber \\
&& \left. + \; a\left(\lambda_0\right){}^2 a\left(\lambda_1\right) a\left(\lambda_2\right) b\left(\lambda_1\right) b\left(\lambda_2\right)+a\left(\lambda_0\right){}^2 a\left(\lambda_1\right){}^2 a\left(\lambda_2\right){}^2 \right. \nonumber \\
&& \left. + \; b\left(\lambda_0\right){}^2 b\left(\lambda_1\right){}^2 b\left(\lambda_2\right){}^2\right] 
\>
which allows a prompt comparison with \eqref{Z3}. By doing so we then find a perfect agreement between \eqref{Z3} and \eqref{Z3r} with respective parameters $\kappa_0$ also included in Table \ref{tab:L3}.

\begin{table}[h]
\caption{Eigenvalues $\Lambda$ for $L=3$.}
\label{tab:L3}
\begin{center}
\begin{tabular}{|c|c|}
\hline
$\kappa_0$ & $\Lambda(\lambda)$ \\
\hline
&  \\[-1em]
$1$ & $\frac{1}{4} \sinh{(\gamma )} \left[-\left(1+\ii \sqrt{3}\right) \cosh{ (2 (\gamma +\lambda ))}+\ii \left(\sqrt{3}+\ii\right) \cosh{ (2 \lambda )}+2\right]$ \\[1mm]
\hline
&  \\[-1em]
$-1$ & $\frac{1}{4} \sinh{ (\gamma )} \left[\left(1-\ii \sqrt{3}\right)  \cosh{ (2 (\gamma +\lambda ))} + \left(1+\ii \sqrt{3}\right) \cosh{ (2 \lambda )}-2\right]$ \\[1mm]
\hline
&  \\[-1em]
$1$ & $\frac{1}{4} \sinh{ (\gamma )} \left[i \left(\sqrt{3} + \ii \right) \cosh{(2 (\gamma +\lambda ))} - \left(1 + \ii \sqrt{3}\right) \cosh{ (2 \lambda )}+2\right]$ \\[1mm]
\hline
&  \\[-1em]
$-1$ & $\frac{1}{4} \sinh{(\gamma )} \left[\left(1 + \ii \sqrt{3}\right) \cosh{(2 (\gamma +\lambda))}+\left(1- \ii \sqrt{3}\right) \cosh{(2 \lambda )}-2\right]$ \\[1mm]
\hline 
&  \\[-1em]
\multirow{2}{*}{$1$} & $\frac{1}{4} \sinh{(\gamma )} \left[\cosh{(2 (\gamma +\lambda))} - 2 \sqrt{2} \sqrt{\cosh{ (2 \gamma )}+7} \sinh{ (\lambda )} \sinh{ (\gamma +\lambda )} \right. $ \\ &  $\left. + \cosh{(2 \gamma )} + \cosh{(2 \lambda)} - 3\right]$ \\
\hline
&  \\[-1em]
\multirow{2}{*}{$-1$} & $-\frac{1}{4} \sinh{ (\gamma )} \left[\cosh{ (2 (\gamma +\lambda))} + 2 \sqrt{2} \sqrt{\cosh{(2 \gamma )}+7} \sinh{ (\lambda )} \sinh{(\gamma +\lambda )} \right.$ \\ & $ \left. + \cosh{(2 \gamma )} + \cosh{(2 \lambda )}-3\right]$ \\
\hline
&  \\[-1em]
\multirow{2}{*}{$1$} & $\frac{1}{4} \sinh{(\gamma )} \left[\cosh{ (2 (\gamma +\lambda))} + 2 \sqrt{2} \sqrt{\cosh{(2 \gamma )} + 7} \sinh{ (\lambda )} \sinh{ (\gamma +\lambda )} \right.$ \\ & $\left. + \cosh{ (2 \gamma )} + \cosh{ (2 \lambda)} - 3\right]$ \\
\hline
&  \\[-1em]
\multirow{2}{*}{$-1$} & $-\frac{1}{4} \sinh{(\gamma )} \left[\cosh{(2 (\gamma +\lambda ))} - 2 \sqrt{2} \sqrt{\cosh{(2 \gamma )} + 7} \sinh{(\lambda )} \sinh{(\gamma +\lambda )} \right.$ \\ & $\left. + \cosh{(2 \gamma )} + \cosh{(2 \lambda)} - 3\right]$ \\
\hline
\end{tabular}
\end{center}
\end{table}

As far as the homogeneous limit is concerned, more precisely the limits $\lambda_j \to \lambda$ and $\mu_j \to \mu$; formula \eqref{Z3} is not so friendly as its counterpart \eqref{Z2} for $L=2$. This is mainly due to the presence of simple poles in the coefficients $M_i^{(2)}$ and $N_{2,1}^{(2)}$ when
$\lambda_j \to \lambda$. Nevertheless, the limit $\mu_j \to \mu$ is still trivial in formula \eqref{Z3}.

\subsection{Case $L=4$} \label{sec:L4}

This case is of particular importance because it reveals the structure we need to consider in order to generalize formulae \eqref{Z2} and \eqref{Z3}
for arbitrary lattice length $L$. Similarly to the previous cases, we start this subsection by writing down explicitly the system of equations \eqref{FN} for $L=4$. Here, however, we split them into two blocks from the start. The first one consists of the equations $n=0$ and $n=1$, namely
\< \label{FN4a}
\Lambda(\lambda_0) \mathcal{F}_0 &=& \mathcal{F}_1 (\lambda_0) \nonumber \\
\Lambda(\lambda_0) \mathcal{F}_1 (\lambda_1) &=& \mathcal{F}_2 (\lambda_0, \lambda_1)  + M_1^{(1)} \; \mathcal{F}_0 \; .
\>
The second block is then formed by the remaining two equations $n=2$ and $n=3$. More precisely, the second block is given by
\< \label{FN4b}
\Lambda(\lambda_0) \mathcal{F}_2 (\lambda_1, \lambda_2) &=& \mathcal{F}_3 (\lambda_0, \lambda_1 , \lambda_2) + M_1^{(2)} \; \mathcal{F}_1 (\lambda_2)  +  M_2^{(2)} \; \mathcal{F}_1 (\lambda_1) + N_{2,1}^{(2)}\; \mathcal{F}_1 (\lambda_0) \nonumber \\
\Lambda(\lambda_0) \mathcal{F}_3 (\lambda_1, \lambda_2, \lambda_3) &=& Z (\lambda_0, \lambda_1 , \lambda_2, \lambda_3) \bar{\mathcal{F}}_0 
+ M_1^{(3)} \; \mathcal{F}_2 (\lambda_2, \lambda_3) + M_2^{(3)} \; \mathcal{F}_2 (\lambda_1, \lambda_3) \nonumber \\
&& + \; M_3^{(3)} \; \mathcal{F}_2 (\lambda_1, \lambda_2) +  N_{2,1}^{(3)} \; \mathcal{F}_2 (\lambda_0, \lambda_3) + N_{3,1}^{(3)} \; \mathcal{F}_2 (\lambda_0, \lambda_2) \nonumber \\
&& + \; N_{3,2}^{(3)} \; \mathcal{F}_2 (\lambda_0, \lambda_1) \; . 
\>

Next we need to identify suitable variables for extending the procedure previously put forward in the cases $L=2$ and $L=3$. For that we turn our attention to the last equation of \eqref{FN4b} and select all the terms in its RHS which does not involve the partition function $Z$. More precisely, we single out the functions $\mathcal{F}_2 (\lambda_i, \lambda_j)$ with indexes on the interval $0 \leq i < j \leq 3$. Each one of those functions can now be written down in terms of $\mathcal{F}_1$ and $\mathcal{F}_0$ through the last equation of \eqref{FN4a}. Similarly, each function $\mathcal{F}_1$ can then be written down in terms of $\mathcal{F}_0$ according to the first equation of \eqref{FN4a}. 
This inspection suggests one can write a closed sub-system of equations involving the variables $\mathcal{F}_2 (\lambda_i, \lambda_j)$ with $0 \leq i < j \leq 3$, $\mathcal{F}_1 (\lambda_i)$ with $1 \leq i \leq 3$ and $\mathcal{F}_0$. The equations required for that are the last equation of the second block \eqref{FN4b} and the whole first block \eqref{FN4a} with permuted spectral parameters $\lambda_k$. Hence, we end up with a system of ten equations for the afore described ten variables. Such system can be solved using Cramer's rule and the solution for $\mathcal{F}_0$ allows us to write
\[ \label{part4}
\Lambda(\lambda_0) \mathcal{F}_3 (\lambda_1, \lambda_2, \lambda_3) - Z (\lambda_0, \lambda_1 , \lambda_2, \lambda_3) \bar{\mathcal{F}}_0 = \mathcal{F}_0 \; \mathrm{det} \left( \mathcal{C}_4 \right)
\]
with $\mathcal{C}_4$ the coefficient matrix of the above described system of equations. The explicit form of the matrix $\mathcal{C}_4$ will not required at this point. 

We then proceed by looking for a suitable expression for the function $\mathcal{F}_3 (\lambda_1, \lambda_2, \lambda_3)$ appearing in \eqref{part4}. Such expression can be found from the solution of another system of equations similar to the one leading to \eqref{part4}. More precisely, the system we need is obtained by using the first equation of \eqref{FN4b}, with renamed spectral parameters $\lambda_i \mapsto \lambda_{i+1}$, instead of the last equation
in \eqref{FN4b}. All the other equations remain the same. The solution of this latter system of equations for $\mathcal{F}_0$ then gives us the expression
\[ \label{part4b}
\mathcal{F}_3 (\lambda_1, \lambda_2, \lambda_3) = \mathcal{F}_0 \; \mathrm{det} \left( \bar{\mathcal{C}}_4 \right) \; ,
\]
with matrix $\bar{\mathcal{C}}_4$ the coefficient matrix of this second system of equations. It is important to remark that $\mathcal{C}_4$ and $\bar{\mathcal{C}}_4$ only differ by the last row. In this way, their determinants can be neatly combined, and from \eqref{part4} and \eqref{part4b} we find
\[ \label{Z4}
Z ( \lambda_0, \lambda_1, \lambda_2, \lambda_3) = \kappa_0 \; \mathrm{det} \left( \mathcal{H}_4 \right) \; 
\]
with $\mathcal{H}_4$ the following $10\times 10$ matrix,
\< \label{H4}
\scalebox{0.75}{$
\begin{pmatrix}
 \Lambda(\lambda_3) & -1 & 0 & 0 & 0 & 0 & 0 & 0 & 0 & 0 \\
 \Lambda(\lambda_2) & 0 & -1 & 0 & 0 & 0 & 0 & 0 & 0 & 0 \\
 \Lambda(\lambda_1) & 0 & 0 & -1 & 0 & 0 & 0 & 0 & 0 & 0 \\
 -M_1^{(1)} (\lambda_2, \lambda_3) & \Lambda(\lambda_2) & 0 & 0 & -1 & 0 & 0 & 0 & 0 & 0 \\
 -M_1^{(1)} (\lambda_1, \lambda_3) & \Lambda(\lambda_1) & 0 & 0 & 0 & -1 & 0 & 0 & 0 & 0 \\
 -M_1^{(1)} (\lambda_1, \lambda_2)  & 0 & \Lambda(\lambda_1) & 0 & 0 & 0 & -1 & 0 & 0 & 0 \\
 -M_1^{(1)} (\lambda_0, \lambda_3)  & \Lambda(\lambda_0) & 0 & 0 & 0 & 0 & 0 & -1 & 0 & 0 \\
 -M_1^{(1)} (\lambda_0, \lambda_2)  & 0 & \Lambda(\lambda_0) & 0 & 0 & 0 & 0 & 0 & -1 & 0 \\
 -M_1^{(1)} (\lambda_0, \lambda_1)  & 0 & 0 & \Lambda(\lambda_0) & 0 & 0 & 0 & 0 & 0 & -1 \\
 0 & -\Lambda(\lambda_0) M_1^{(2)} & -\Lambda(\lambda_0) M_2^{(2)} & -\Lambda(\lambda_0) N_{2,1}^{(2)} & \Lambda(\lambda_0) \Lambda(\lambda_1) - M_1^{(3)} & - M_2^{(3)} & - M_3^{(3)} & - N_{2,1}^{(3)} & - N_{3,1}^{(3)} & - N_{3,2}^{(3)} \\
\end{pmatrix} $} \; . \nonumber \\
\>
As for the dependence on the spectral parameters, in \eqref{H4} we have used the conventions $M_i^{(2)} = M_i^{(2)} (\lambda_1, \lambda_2, \lambda_3)$,
$N_{2,1}^{(2)} = N_{2,1}^{(2)} (\lambda_1, \lambda_2, \lambda_3)$, $M_i^{(3)} = M_i^{(3)} (\lambda_0, \lambda_1, \lambda_2, \lambda_3)$ and $N_{j,i}^{(3)} = N_{j,i}^{(3)} (\lambda_0, \lambda_1, \lambda_2, \lambda_3)$.

\begin{rema}
Similarly to the case $L=3$, the partial homogeneous limit $\mu_j \to \mu$ is trivial in formula \eqref{Z4}. However, some matrix entries of $\mathcal{H}_4$
exhibit simple poles when $\lambda_j \to \lambda$ and this makes the analysis of the complete homogeneous limit of formula \eqref{Z4} more involving. 
\end{rema}

\subsection{General case} \label{sec:LL}

As for arbitrary values of the lattice length $L$ we will simply proceed along the lines described in \Secref{sec:L4}.
We start by separating the system of equations \eqref{FN} into two blocks. In the first block we put the equations
\< \label{FNa}
\Lambda (\lambda_0) \mathcal{F}_n (X^{1,n}) &=& \mathcal{F}_{n+1} (X^{0,n}) + \sum_{1 \leq i \leq n} M_i^{(n)} \mathcal{F}_{n-1} (X^{1,n}_i) \nonumber\\
&& + \; \sum_{1 \leq i < j \leq n} N_{j,i}^{(n)} \mathcal{F}_{n-1} (X^{0,n}_{i,j})
\>
for  $0 \leq n \leq L-3$. The remaining equations, namely \eqref{FN} with $n=L-2$ and $n=L-1$, will constitute the second block. More precisely, the second contains the relations
\< \label{FNb}
\Lambda (\lambda_0) \mathcal{F}_{L-2} (X^{1,L-2}) &=& \mathcal{F}_{L-1} (X^{0,L-2}) + \sum_{1 \leq i \leq L-2} M_i^{(L-2)} \mathcal{F}_{L-3} (X^{1,L-2}_i) \nonumber\\
&& + \; \sum_{1 \leq i < j \leq L-2} N_{j,i}^{(L-2)} \mathcal{F}_{L-3} (X^{0,L-2}_{i,j}) \nonumber \\
\Lambda (\lambda_0) \mathcal{F}_{L-1} (X^{1,L-1}) &=& Z (X^{0,L-1}) \bar{\mathcal{F}}_0 + \sum_{1 \leq i \leq L-1} M_i^{(L-1)} \mathcal{F}_{L-2} (X^{1,L-1}_i) \nonumber\\
&& + \; \sum_{1 \leq i < j \leq L-1} N_{j,i}^{(L-1)} \mathcal{F}_{L-2} (X^{0,L-1}_{i,j}) \; .
\>
We will also need to extend Definition \ref{VAR} in order to proceed.

\begin{defi}
Write  $X^{a,b}_{i_1, i_2, \dots, i_m} \coloneqq X^{a,b} \backslash \{ \lambda_{i_1} , \lambda_{i_2} , \dots , \lambda_{i_m} \}$ for a reduced set of symmetric variables.
\end{defi}

The next step is to recast \eqref{FNa} and \eqref{FNb} in a convenient matricial form. For that we use as variable the vector with entries
\[ \label{Va}
\mathcal{F}_{L-m} (X^{0,L-1}_{i_1 ,i_2 , \dots, i_m}) \quad \text{with} \quad 0 \leq i_1 < i_2 < \dots < i_m \leq L-1 \quad \text{if} \quad m \in 2 \Z \; ,
\]
and
\[ \label{Vb}
\mathcal{F}_{L-m} (X^{1,L-1}_{i_1 ,i_2 , \dots, i_{m-1}}) \quad \text{with} \quad 1 \leq i_1 < i_2 < \dots < i_{m-1} \leq L-1 \quad \text{if} \quad m \in 2 \Z +1 \; .
\]
The index $m$ in \eqref{Va} and \eqref{Vb} runs in decreasing order as $m = L, L-1, \dots, 2$.

\begin{rema}
Formulae \eqref{Va} and \eqref{Vb} gives rise to a vector with $3 \cdot 2^{L-2} - 2$ components.
\end{rema}

\begin{lem} \label{Zlema}
The partition function $Z$ can be written as 
\[ \label{ZL}
Z(\lambda_0 , \lambda_1, \dots , \lambda_{L-1}) = \kappa_0 \; \mathrm{det} \left( \mathcal{H}_L \right)
\]
with $\mathcal{H}_L$ the following tridiagonal block matrix 
\< \label{HL}
\mathcal{H}_L  = \begin{pmatrix}
\mathcal{A}_{0}^{(L)} & \mathcal{A}_{+}^{(L)} & & & & & \\
\mathcal{A}_{-}^{(L-1)} & \mathcal{A}_{0}^{(L-1)} & \mathcal{A}_{+}^{(L-1)} & & & & \\
& \ddots & \ddots & \ddots &  &  & \\
& & \mathcal{A}_{-}^{(m)} & \mathcal{A}_{0}^{(m)} & \mathcal{A}_{+}^{(m)} & & \\
& &  &\ddots & \ddots & \ddots &    \\
& & & & \mathcal{A}_{-}^{(3)} & \mathcal{A}_{0}^{(3)} &  \mathcal{A}_{+}^{(3)} \\
& &  &  & & \mathcal{A}_{-}^{(2)} & \mathcal{A}_{0}^{(2)} \\
\end{pmatrix} \; .
\>
As for $m=3,4, \dots , L$ and $\mathbb{1}_d$ the $d \times d$ identity matrix we have $\mathcal{A}_{+}^{(m)} \coloneqq - \mathbb{1}_{[m]}$ with 
\<
[m] \coloneqq \begin{cases}
\frac{(L-1)!}{(m-2)! (L-m+1)!} \qquad \text{for} \quad m \in 2 \Z \\
\frac{L!}{(m-1)! (L-m+1)!} \qquad \text{for} \quad m \in 2 \Z +1 
\end{cases} \; .
\>
In its turn, $\mathcal{A}_{0}^{(m)}$ for  $m > 2$ is a $[m] \times [m+1]$ matrix given by
\<
\mathcal{A}_{0}^{(m)} = \begin{pmatrix}
\mathcal{W}_{m}^{(m)} & 0 & \dots & 0 \\
\mathcal{W}_{m-1}^{(m)} & 0 & \dots & 0 \\
\vdots & 0 & \dots & 0 \\
\mathcal{W}_{3}^{(m)} & 0 & \dots & 0 \\
\mathcal{W}_{2}^{(m)} & 0 & \dots & 0 
\end{pmatrix} \quad \text{for} \; m \; \text{even; and by} \;\;  
\mathcal{A}_{0}^{(m)} = \begin{pmatrix}
\mathcal{W}_{m}^{(m)} & 0 & \dots & 0 \\
\mathcal{W}_{m-1}^{(m)} & 0 & \dots & 0 \\
\vdots & 0 & \dots & 0 \\
\mathcal{W}_{2}^{(m)} & 0 & \dots & 0 \\
\mathcal{W}_{1}^{(m)} & 0 & \dots & 0 
\end{pmatrix} \nonumber \\
\>
for odd values of $m$. In both cases $\mathcal{W}_k^{(m)} \coloneqq \Lambda(\lambda_{k-1}) \mathbb{1}_{\{ m; k \}}$ with 
\[
\{ m; k \} \coloneqq \frac{(L-k)!}{(m-k)! (L-m)!} \; .
\]

As for $m>2$, the matrices $\mathcal{A}_{-}^{(m)}$ have dimension $[m] \times [m+2]$. Their entries for $m \in 2 \Z$ read 
\< \label{am}
&& \left( \mathcal{A}_{-}^{(m)} \right)_{r_{i_1, i_2, \dots , i_{m-2}}, \; s_{j_1, j_2, \dots , j_{m}}} = - \prod_{t=1}^{m-2} \delta_{i_t , j_t} \nonumber \\
&& \qquad\qquad \times
\begin{cases}
M_k^{(L-m)} \left( \mathcal{O} (X^{1,L-1}_{i_1 ,i_2 , \dots, i_{m-2}}) \right) \quad \text{if} \;\; \mathcal{O} (X^{1,L-1}_{i_1 ,i_2 , \dots, i_{m-2}})_{1,k+1} = \mathcal{O} ( X^{1,L-1}_{j_1 ,j_2 , \dots, j_{m}}) \nonumber \\
N_{l,k}^{(L-m)} \left( \mathcal{O} (X^{1,L-1}_{i_1 ,i_2 , \dots, i_{m-2}}) \right) \quad \text{if} \;\; \mathcal{O} (X^{1,L-1}_{i_1 ,i_2 , \dots, i_{m-2}})_{k+1, l+1} = \mathcal{O} ( X^{1,L-1}_{j_1 ,j_2 , \dots, j_{m}}) \nonumber \\
\end{cases} \; .
\>
In \eqref{am} $r_{i_1, i_2, \dots , i_{m-2}} \colon \Z^{\times (m-2)} \to \Z$ corresponds to the integer assigned to the ordering $1 \leq i_1 < i_2 < \dots < i_{m-2} \leq L-1$; and similarly $s_{j_1, j_2, \dots , j_{m}} \colon \Z^{\times m} \to \Z$ assigns an integer to the ordering $1 \leq j_1 < j_2 < \dots < j_m \leq L-1$.

Next we assign the integers $\bar{r}_{i_1, i_2, \dots , i_{m-1}} \colon \Z^{\times (m-1)} \to \Z$ to the ordering $0 \leq i_1 < i_2 < \dots < i_{m-1} \leq L-1$ and
$\bar{s}_{j_1, j_2, \dots , j_{m+1}} \colon \Z^{\times (m+1)} \to \Z$ to $0 \leq j_1 < j_2 < \dots < j_{m+1} \leq L-1$.
The matrices $\mathcal{A}_{-}^{(m)}$ for $m \in 2 \Z+1$ and $m>2$ then have entries
\< \label{am}
&& \left( \mathcal{A}_{-}^{(m)} \right)_{\bar{r}_{i_1, i_2, \dots , i_{m-1}}, \; \bar{s}_{j_1, j_2, \dots , j_{m+1}}} = - \prod_{t=1}^{m-1} \delta_{i_t , j_t} \nonumber \\
&& \qquad\qquad \times
\begin{cases}
M_k^{(L-m)} \left( \mathcal{O} (X^{0,L-1}_{i_1 ,i_2 , \dots, i_{m-1}}) \right) \quad \text{if} \;\; \mathcal{O} (X^{0,L-1}_{i_1 ,i_2 , \dots, i_{m-1}})_{1,k+1} = \mathcal{O} ( X^{0,L-1}_{j_1 ,j_2 , \dots, j_{m+1}}) \nonumber \\
N_{l,k}^{(L-m)} \left( \mathcal{O} (X^{0,L-1}_{i_1 ,i_2 , \dots, i_{m-1}}) \right) \quad \text{if} \;\; \mathcal{O} (X^{0,L-1}_{i_1 ,i_2 , \dots, i_{m-1}})_{k+1, l+1} = \mathcal{O} ( X^{0,L-1}_{j_1 ,j_2 , \dots, j_{m+1}}) \nonumber \\
\end{cases} \; .
\>
At last we have 
\< \label{HLf}
&&\mathcal{A}_{-}^{(2)} = - \Lambda(\lambda_0) \left( M_1^{(L-2)} , M_2^{(L-2)} , \dots , M_{L-2}^{(L-2)}, N_{2,1}^{(L-2)}, N_{3,1}^{(L-2)}, \dots, \right. \nonumber \\
&& \qquad\qquad\qquad\qquad\qquad\qquad\qquad\qquad\quad \left. \dots, N_{L-2,1}^{(L-2)} , N_{3,2}^{(L-2)}, \dots, N_{L-2,2}^{(L-2)} , \dots , N_{L-2,L-3}^{(L-2)}  \right) \nonumber \\
&&\mathcal{A}_{0}^{(2)} = - \left( M_1^{(L-1)} - \Lambda(\lambda_0) \Lambda(\lambda_1) , M_2^{(L-1)} , \dots , M_{L-1}^{(L-1)}, N_{2,1}^{(L-1)}, N_{3,1}^{(L-1)}, \dots \right. \nonumber \\
&& \qquad\qquad\qquad\qquad\qquad\qquad\qquad\qquad\quad \left. \dots, N_{L-1,1}^{(L-1)} , N_{3,2}^{(L-1)}, \dots, N_{L-1,2}^{(L-1)} , \dots , N_{L-1,L-2}^{(L-1)}  \right) \nonumber \\
\>
with 
\begin{align} 
M_i^{(L-2)} =& M_i^{(L-2)} (\lambda_1, \lambda_2 , \dots , \lambda_{L-1}) & N_{j,i}^{(L-2)} &= N_{j,i}^{(L-2)} (\lambda_1, \lambda_2 , \dots , \lambda_{L-1}) \nonumber \\
M_i^{(L-1)} =& M_i^{(L-1)} (\lambda_0, \lambda_1 , \dots , \lambda_{L-1}) & N_{j,i}^{(L-1)} &= N_{j,i}^{(L-1)} (\lambda_0, \lambda_1 , \dots , \lambda_{L-1}) \; . \nonumber 
\end{align} 
\end{lem}

\begin{proof}
Formula \eqref{ZL} follows from a straightforward generalization of the analysis performed in \Secref{sec:L4} for the case $L=4$. As for generic $L$ we then start by considering the system of equations formed by \eqref{FNa}, with different spectral parameters, and the last equation of \eqref{FNb}. However, in order to completely formulate our system of equations we still need to declare the variables under consideration. As for that we take \eqref{Va} and \eqref{Vb} and 
allocate them as entries of a vector $\vec{\mathcal{F}}$. In this way, one can write the above described system of linear equations as $\mathcal{C}_L \; \vec{\mathcal{F}} = \vec{\mathcal{I}}$,
with $\mathcal{C}_L$ the corresponding matrix of coefficients and 
\[
\vec{\mathcal{I}} \coloneqq \begin{pmatrix}
0 \\ 0 \\ \vdots \\ 0 \\ \Lambda (\lambda_0) \mathcal{F}_{L-1} (X^{1,L-1}) - Z (X^{0,L-1}) \bar{\mathcal{F}}_0
\end{pmatrix} \; .
\]
Next we use Cramer's rule to solve our system of equations for $\mathcal{F}_0$. By doing so we find
\[
\mathcal{F}_0 = \frac{\mathrm{det} \left( \mathcal{C}_L^{(1)} \right) }{\mathrm{det} \left( \mathcal{C}_L \right) }
\]
with $\mathcal{C}_L^{(1)}$ the matrix obtained by replacing the first column of $\mathcal{C}_L$ by the vector $\vec{\mathcal{I}}$. Fortunately, the evaluation of
$\mathrm{det} \left( \mathcal{C}_L^{(1)} \right)$ is trivial and this allows us to write
\[ \label{P1}
\Lambda (\lambda_0) \mathcal{F}_{L-1} (X^{1,L-1}) - Z (X^{0,L-1}) \bar{\mathcal{F}}_0 = \mathcal{F}_0 \; \mathrm{det} \left( \mathcal{C}_L \right) \; .
\]

The next step is to obtain a suitable expression for $\mathcal{F}_{L-1} (X^{1,L-1})$. With that goal in mind we consider a second system of equations formed by \eqref{FNa} and the first equation of \eqref{FNb} with shifted spectral parameters $\lambda_i \mapsto \lambda_{i+1}$. It is important to remark that this second system only differs from the first one by the last equation. We then write $\bar{\mathcal{C}}_L$ for the matrix of coefficients associated to this second system of equations and, by repeating the exact same procedure employed for the first system, we find
\[ \label{P2}
\mathcal{F}_{L-1} (X^{1,L-1}) = \mathcal{F}_0 \; \mathrm{det} \left( \bar{\mathcal{C}}_L \right) \; .
\]
The matrices $\mathcal{C}_L$ and $\bar{\mathcal{C}}_L$ only differ by the last row and this allows us to combine their determinants in a neat way.
In this way, using \eqref{P1} and \eqref{P2} we find \eqref{ZL} with $\kappa_0 = \mathcal{F}_0 / \bar{\mathcal{F}}_0$ and matrix $\mathcal{H}_L$
defined in \eqref{HL}-\eqref{HLf}.
\end{proof} 

\begin{rema}
Formula \eqref{ZL} is valid for $L > 2$ and one can clearly see that the representation \eqref{Z2} for the case $L=2$ is not included in \eqref{ZL}. This is due to the fact that our system of equations in the case $L=2$, namely \eqref{FN2}, can not be distributed into the two blocks \eqref{FNa} and \eqref{FNb}. The latter structure is ultimately responsible for the general formula \eqref{ZL}.
\end{rema}

The representation for the partition function $Z$ described in Lemma \ref{Zlema} exhibits some unusual features whose consequences will be addressed in more details in 
\Secref{sec:LAM}. However, here it is worth mentioning  that such representation is valid for any eigenvalue $\Lambda$ of the transfer matrix of the anti-periodic six-vertex model. Hence, it comprises $2^L$ determinantal representations and it can be regarded as an \emph{invariant} in the space of the transfer matrix's eigenvalues.

\subsection{The coefficient $\kappa_0$} \label{sec:K0}

Given the goals of this paper, formula \eqref{ZL} is not complete until we determine the coefficient $\kappa_0$. The latter has been defined as the ratio
$\mathcal{F}_0 /\bar{\mathcal{F}}_0$ and we recall that $\mathcal{F}_0$ and $\bar{\mathcal{F}}_0$ are projections of the transfer matrix eigenvector associated to the particular
eigenvalue $\Lambda$ entering \eqref{ZL} on the $\mathfrak{sl}_2$ highest and lowest weight vectors respectively.
More precisely, as described in \cite{Galleas_Twists}, we have $\mathcal{F}_0 = \braket{\Psi}{0}$ and $\bar{\mathcal{F}}_0 = \braket{\Psi}{\bar{0}}$. The vectors 
$\ket{0}$ and $\ket{\bar{0}}$ are respectively the $\mathfrak{sl}_2$ highest and lowest weight vectors, whilst $\ket{\Psi}$ solves the eigenvalue equation
$T(\lambda) \ket{\Psi} = \Lambda(\lambda) \ket{\Psi}$. 
In this way, $\kappa_0$ is intrinsically associated to the particular eigenvalue $\Lambda$ we choose in \eqref{ZL}. Also, it is important to remark that the appearance of such overall multiplicative parameter is due to the fact that our system of equations, namely \eqref{FN}, is linear in $\mathcal{F}_n$.
 
The coefficient $\kappa_0$ does not depend on the spectral parameters $\lambda_j$ and, therefore, it can be fixed by evaluating the partition function $Z$ at a particular value of its spectral parameters. At first this seems to be a daunting task but fortunately there exist special points where such computation is doable. For instance, from 
\eqref{ZL} one finds
\[
Z(\mu_1, \mu_2, \dots, \mu_L) = \kappa_0 \prod_{j=1}^L \Lambda(\mu_j) \; ,
\]
which leads us to the problem of computing the partition function $Z(\lambda_1, \lambda_2, \dots , \lambda_L)$ at the particular specialization $\lambda_j = \mu_j$.
Now, recalling $Z$ is a sum of products of statistical weights over six-vertex model's configurations (with domain-wall boundaries), one can readily see that there is only one possible (non-vanishing) configuration for this particular specialization of spectral parameters. More precisely, we have
\[
Z(\mu_1, \mu_2, \dots, \mu_L) = c^L \prod_{1 \leq i < j \leq L} a(\mu_i - \mu_j) a(\mu_j - \mu_i)
\]
which gives us
\[
\kappa_0 =  c^L \frac{\displaystyle  \prod_{1 \leq i < j \leq L} a(\mu_i - \mu_j) a(\mu_j - \mu_i)}{\displaystyle  \prod_{1 \leq i \leq L} \Lambda(\mu_i) } \; .
\]

\section{Consequences for $\Lambda$} \label{sec:LAM}

In this section we intend to comment on the consequences of formula \eqref{ZL} for the transfer matrix's eigenvalues $\Lambda$. Here, we do not intend to exploit such consequences thoroughly but rather point out a few ramifications of formula \eqref{ZL} when combined with well known results for the partition function $Z$.

\subsection{Functional equations for $\Lambda$} \label{sec:EQ1}

As previously mentioned, there exist several representations for the partition function $Z$. For instance, it can be written as the well known Izergin-Korepin determinant or the continuous determinantal representations of \cite{Galleas_2016b}. In this way, the combination of the aforementioned representations with formula \eqref{ZL} leaves us with a functional equation for the eigenvalue $\Lambda$. Such equation is similar to the one obtained in \cite{Galleas_2017} for the periodic 
six-vertex model; however, it is not written as the vanishing of a determinant. Nevertheless, it can still be regarded as an inhomogeneous version of the equation derived in \cite{Galleas_2017}. The inhomogeneity term is then the partition function $Z$. 

\subsection{Relation between two eigenvalues} \label{sec:LamLam}

Let $\Lambda$ and $\bar{\Lambda}$ be two distinct eigenvalues of the anti-periodic six-vertex model's transfer matrix $T$. We then attach $\kappa_0$ and $\mathcal{H}_L$ to
the eigenvalue $\Lambda$; whilst $\bar{\kappa}_0$ and $\bar{\mathcal{H}}_L$ is associated to $\bar{\Lambda}$. Since $Z$, given by formula \eqref{ZL}, is an invariant
in the space of the transfer matrix's eigenvalues, we are consequently left with the relation
\[ \label{bbar}
\kappa_0 \; \mathrm{det} \left( \mathcal{H}_L \right) = \bar{\kappa}_0 \; \mathrm{det} \left( \bar{\mathcal{H}}_L \right) \; .
\]
Eq. \eqref{bbar} then establishes a functional relation between any two eigenvalues of the transfer matrix $T$.

\subsection{Recurrence relation for $\Lambda$} \label{sec:REC}

One important result obtained by Korepin in \cite{Korepin_1982} is a recurrence relation characterizing the partition function $Z$. In order to precise Korepin's relation we then write $Z_L (X^{1,L})$ with subscript $L$ for the partition function $Z (X^{1,L})$ on a $L \times L$ square lattice. Moreover, we also write $\Lambda_L$ for the eigenvalue $\Lambda$ in order to emphasize such eigenvalues are also defined for a fixed lattice length $L$. In this way, Korepin's recurrence relation reads
\< \label{RR}
\left. Z_L (X^{1,L}) \right|_{\lambda_i = \mu_j - \gamma} = - c \prod_{l,m = 1}^L b(\lambda_l - \mu_i) b(\lambda_m - \mu_j) Z_{L-1} (X_i^{1,L}) \; .
\>
Now, by substituting formula \eqref{ZL} in \eqref{RR}, we immediately obtain a recurrence relation involving the eigenvalues $\Lambda_L$ and $\Lambda_{L-1}$.

\section{Concluding remarks} \label{sec:CONCL}

Formula \eqref{ZL} is the main result of this work and it is based on a refinement of the analysis previously presented in \cite{Galleas_Twists}. In particular, the relation between six-vertex models with domain-wall boundaries and non-diagonal twisted boundaries had been already put forward in \cite{Galleas_Twists}. 
However, the determinantal formula obtained here makes this relation more concise, offering new perspectives for the study of both systems.

Our present analysis is based on the role played by the \emph{symmetric group} in the study of functional equations originated from the 
Algebraic-Functional method \cite{Galleas_2008, Galleas_2010}. The prominent role of the symmetric group within that method was previously unveiled in \cite{Galleas_2016a, Galleas_2016b, Galleas_2016c},
and here we have extended the symmetric group analysis for the equations derived in \cite{Galleas_Twists}. 
The equations \eqref{FN} considered here are clearly more involving than the ones studied in \cite{Galleas_2016a, Galleas_2016b, Galleas_2016c}, but the main ideas used in the aforementioned works can be straightforwardly adapted to the present case.

Both six-vertex models with domain-wall boundaries and non-diagonal twisted boundaries had been previously studied in the literature and a considerable amount of information is available  for both systems. In this way, formula \eqref{ZL} offers the possibility of using results of the six-vertex model with domain-wall boundaries to the case with non-diagonal twists and vice-versa. Some of these possibilities are discussed in \Secref{sec:LAM} but they probably do not exhaust the applications of
formula \eqref{ZL} along these lines.

As for the determinantal representation \eqref{ZL}, we can summarize some of its main features:
\begin{enumerate}[label=(\roman*)]
\item it totals $2^L$ representations depending on the transfer matrix eigenvalue $\Lambda$ we use in the RHS; \\
\item it is an invariant on the spectrum of the anti-periodic transfer matrix $T$. That is, it does not depend on the particular eigenvalue $\Lambda$ we use in the RHS; \\
\item it provides a relation between two distinct eigenvalues, i.e. $\Lambda$ and $\bar{\Lambda}$; \\
\item previous results for $Z$ implies  a functional equation for the eigenvalues $\Lambda$; \\
\item the known recurrence relations for $Z$ are immediately translated into recurrence relations for the eigenvalues $\Lambda$.
\end{enumerate}
In the present work we have not explored any of the above mentioned directions of research but we hope to report on such possibilities in future publications.

\bibliographystyle{alpha}
\bibliography{references1}

%
%
%
%
%
%

\end{document}